\documentclass[a4paper,onecolumn,unpublished]{quantumarticle}
\pdfoutput=1
\usepackage{todonotes}
\usepackage{graphicx}
\usepackage{physics}
\usepackage{tcolorbox}
\usepackage{amsthm}

\usepackage{amssymb}
\usepackage{cancel}
\usepackage{supertabular}
\usepackage{hyperref}
\usepackage{cleveref}
\usepackage{algorithm}
\usepackage{algpseudocode}
\algrenewcommand{\algorithmicensure}{\textbf{Output:}}

\usepackage{tikz}
\usetikzlibrary{arrows.meta,decorations.pathmorphing}
\usepackage[utf8]{inputenc}
\DeclareUnicodeCharacter{2217}{*}

\newtheorem{definition}{Definition}
\newtheorem{theorem}{Theorem}
\newtheorem{lemma}{Lemma}

\newtheorem{remark}{Remark}
\newtheorem{proposition}{Proposition}
\newtheorem{note}{Note}

\begin{document}

\title{Quantum ring all-reduce: communication and privacy advantages for distributed learning}

\author{María Gragera Garcés}
\email{m.gragera.garces@ed.ac.uk}
\affiliation{University of Edinburgh, Edinburgh, United Kingdom}
\orcid{0009-0000-9018-7445}

\author{Lirandë Pira}
\email{lpira@nus.edu.sg}
\affiliation{Centre for Quantum Technologies, National University of Singapore, Singapore}
\orcid{0000-0001-7393-2640}

\maketitle

\begin{abstract}
Machine learning models have scaled to unprecedented sizes, making training across distributed devices the de facto standard in the field. 
In this work, we explore how quantum communications can make distributed training both more communication-efficient and information-theoretically private, for both classical and quantum learning models.
Ring all-reduce is the foundational communication primitive for large-scale distributed training. We present a quantum version that reduces per-link \emph{online} communication by a provably optimal factor of two (a saving on the synchronisation critical path; total lifetime channel use is unchanged) using pre-shared entanglement and superdense coding, without requiring the learning model or gradient computation to change. Beyond bandwidth, the primitive enables privacy guarantees that are information-theoretically impossible for any classical protocol, achieving composable $\varepsilon$-secure aggregation, via verified entanglement, at a $2\times$ overhead in GHZ copies. 
Our hybrid quantum-classical communication architecture yields simultaneous communication and security advantages for large scale distributed training, regardless of whether the learning itself is quantum or classical. Finally, we characterise quantum advantages in gradient conflict detection for server-to-client communication under bandwidth constraints, a setting that arises after ring all-reduce is completed, when full gradient broadcast to external clients is infeasible due to bandwidth constraints. Two
variants of the problem admit qualitatively different separations. For margin-based alignment testing (\textsc{GapIP}$_\tau$), the quantum advantage is quadratic in the margin parameter: $\widetilde{O}(\tau^{-1}\log P)$ qubits
versus $\widetilde{O}(\min(\tau^{-2},P))$ bits classically. For sign-consistency auditing against a private parameter matching (\textsc{TieAudit}$_\varepsilon$), the advantage represents an exponential separation in communication complexity: classical one-way communication requires $\Omega(\sqrt{P})$ bits whereas $O(\varepsilon^{-2}\log P)$ qubits suffice. 
\end{abstract}

\newpage
\tableofcontents
\newpage

\section{Introduction}

Training a competitive machine learning model today requires processing datasets of billions of samples, optimising hundreds of billions of parameters, and executing arithmetic operations that would take enormous resources on a single processor. The gap between what a single device can deliver and what training demands has grown faster than hardware alone can close. As model sizes, datasets, and computational demands continue to grow, training is performed across clusters of devices that must repeatedly communicate in order to perform as a single learning system~\cite{bennun2020demystifying}. This suggests that industry-scale machine learning is by default distributed, and thus, modern distributed training is dominated by gradient communication.

Ring all-reduce is a standard communication primitive for distributed training at scale~\cite{patarasuk2009bandwidth}. This algorithm works by arranging $N$ workers in a ring, each holding a local gradient $g_j\in\mathbb{R}^P$ computed on its local data; in each round every worker exchanges one message with each ring neighbour, and after $2(N-1)$ rounds every worker holds the global mean $\bar{g}=\frac{1}{N}\sum_j g_j$ to $b$ bits of precision per coordinate ($Pb$ bits total). The primitive underpins large-scale machine learning, and more specifically deep learning systems across industry, implemented in frameworks such as Horovod~\cite{sergeev2018horovod}, PyTorch\,DDP~\cite{li2020pytorch}, and NCCL (NVIDIA Collective Communications Library).

In this work we explore how quantum technologies can reduce the resources necessary to perform this primitive. One principle organises what follows: the size of the quantum advantage is governed by \emph{what one asks of the all-reduce aggregate}. Transporting its value buys only a Holevo-bounded constant factor; certifying how it was produced buys a qualitative privacy separation, classically impossible at any communication
cost; and querying a privately chosen relation over it can buy an exponential separation. These regimes coexist within a single GHZ-based protocol, leaving the learning model untouched. The first of them, the bandwidth saving, is a baseline rather than the headline. Bandwidth saving relocates online communication off the critical path but does not reduce total channel use (made more precise in \Cref{sec:protocol}), and it is the substrate on which the privacy and conflict-detection results build. The main contributions of this work are threefold.

First, we propose a superdense-coding strategy that exploits pre-shared EPR pairs to replace each classical message of $m$ bits with $\lceil m/2\rceil$ qubits, reducing the online communication across every ring link by exactly a factor of two. This is provably optimal: a cut argument shows at least $Pb - O(\log N)$ classical bits of information must cross every cut in the online phase of the ring, and the entanglement-assisted capacity of a qubit channel is at most two classical bits even with unlimited entanglement, so the factor-of-two is the ceiling of what entanglement can buy. The round structure and latency are unchanged: superdense coding reduces the online communication \emph{volume} rather than the number of rounds, so the quantum protocol blocks for the same $2(N-1)$ post-barrier rounds as the classical schedule. No protocol confined to the ring---classical or quantum, with any pre-shared entanglement---can complete an all-reduce in fewer than the $\lfloor N/2\rfloor$ rounds set by the ring's diameter.

Second, we present the security guarantees a quantum ring all-reduce offers. In a model where every channel (including the setup phase) is wiretapped, classical protocols cannot achieve information-theoretically private aggregation at \emph{any} communication cost, while our Verified GHZ Aggregation (VGA) protocol achieves it at a
$2\times$ overhead in GHZ copies, regardless of whether the learning itself is quantum or classical.

Third, we study two variants of gradient conflict detection for communication-limited settings where broadcasting full gradient vectors is infeasible. In distributed training, the server often broadcasts a compressed global signal that each client must interpret against its own private model structure: the compressed message alone does not tell the client whether the proposed update is compatible with its internal constraints. In the first variant (\textsc{GapIP}), two workers check peer-to-peer whether their unit-normalised gradients are aligned or conflicting up to a margin $\tau$, without exchanging full vectors: the quantum communication cost is $\widetilde{O}(\tau^{-1}\log P)$ versus $\widetilde{O}(\min(\tau^{-2},P))$ classically, a tight quadratic separation. The second variant, \textsc{TieAudit}, arises in the signSGD setting where the server's compressed global signal is the sign-gradient $s\in\{\pm1\}^P$, one bit per parameter. Each client holds a private matching $M_k$ of parameter-index pairs that are structurally
coupled in its local model (via weight sharing, factorisation, or quantisation); the matching acts as a private key for checking whether the global update direction is internally consistent.
Concretely, the client wants to estimate what fraction of its tied pairs $(i,j)\in M_k$ receive conflicting global signs $s_i\neq s_j$: a high fraction signals that the server's update is structurally incompatible with the client's model and should be locally corrected.
The quantum advantage here is exponential: classical one-way communication requires $\Omega(\sqrt{P})$ bits regardless of approximation accuracy, while $O(\varepsilon^{-2}\log P)$ qubits suffice.

The quantum-communication results we draw on here --- such as superdense coding, GHZ-based metrology, the impossibility of secret-key agreement from independent randomness, and Boolean Hidden Matching --- are each individually standard. Our aim is not to promote any particular approach, but to identify the communication layer of distributed learning, a setting of clear practical importance, as a regime in which these results yield concrete advantages. These advantages arise from the communication layer alone: unlike Gilboa et al.~\cite{gilboa2024exponential}, whose exponential separation stems from the nodes holding quantum data, our inputs and outputs can be entirely classical, as the magnitude of our advantage claims is governed by the type of question asked of the ring all-reduce aggregate. Moving the aggregate's value buys a Holevo-bounded constant factor; certifying the setup that produced it buys a qualitative privacy separation, impossible to achieve classically at any communication cost; and extracting a privately chosen relation over it can buy an exponential separation. A single GHZ-based communication protocol spans all three regimes while leaving the learning model untouched: this is the central message of our work.

This paper is organised as follows. 
\Cref{sec:related} presents related work including the ring all-reduce architecture and existing literature regarding quantum advantage in distributed machine learning.
\Cref{sec:protocol} introduces the quantum ring all-reduce primitive in both its raw and GHZ phase encoding forms.
\Cref{sec:security} demonstrates the security advantages of a GHZ phase encoding ring all-reduce architecture protocol.
Finally, \Cref{sec:tieaudit} presents the gradient conflict settings that might arise from ring all-reduce and how quantum architectures can enable quadratic and exponential advantages.
Conclusions and future outlooks can be found in \Cref{sec:conclusion}.

\section{Related work}\label{sec:related}

\subsection{Gradient compression and ring all-reduce}
Ring all-reduce is the bandwidth-optimal topology for gradient aggregation in data-parallel training~\cite{patarasuk2009bandwidth}, and its implementation in Horovod~\cite{sergeev2018horovod} and PyTorch\,DDP~\cite{li2020pytorch} makes it the de facto standard for large-scale deep learning. Gradient compression methods, including signSGD~\cite{bernstein2018signsgd} and its fault-tolerant majority-vote variant~\cite{bernstein2019signsgd}, reduce per-link bandwidth
at the cost of introducing sign-aggregation artifacts; \textsc{TieAudit} formalises one
such artifact as a detectable relational statistic.
The gradient conflict literature~\cite{yu2020gradient,karimireddy2020scaffold} motivates
\textsc{GapIP}: both PCGrad and SCAFFOLD identify conflicting updates as a source of convergence degradation, but neither addresses the communication cost of detecting conflicts under bandwidth constraints.

\subsection{Quantum advantages in distributed machine learning}
Gilboa et al.~\cite{gilboa2024exponential} show an exponential quantum communication advantage for distributed inference when nodes hold \emph{quantum} data; the advantage follows from the exponential classical simulation cost of quantum states rather than from
the communication primitive itself. Our work is complementary: inputs are classical gradient vectors and the quantum advantage
arises entirely from the communication layer (superdense coding for bandwidth, phase encoding aggregation for security, and one-way quantum messages for gradient conflict detection). Doosti et al.~\cite{doosti2026distributed} study distributed quantum state certification under one-way communication constraints and obtain a quadratic quantum-over-classical advantage via quantum random compression. Their setting shares structural parallels with our \textsc{GapIP} and \textsc{TieAudit} problems, both exploit one-way quantum advantages in relational statistics, but addresses state certification to a central node rather than gradient aggregation for learning and conflict detection in a ring. Pira and Ferrie~\cite{pira2023invitation} survey distributed quantum neural networks, focusing on architectures rather than communication primitives, and identify quantum generalisations of collective communication protocols as an open challenge.

\subsection{Quantum secure aggregation}
Classical private aggregation relies on pairwise masking over trusted side channels; quantum proposals extend this via gradient
hiding~\cite{li2024privacy}, blind quantum computation~\cite{li2024blind,zhang2022federated},
and Byzantine-robust quantum aggregation~\cite{nath2026cqsa}.
A common assumption across this line is that entanglement or keys are distributed over a channel not visible to the adversary.
\Cref{thm:cimposs} establishes that the trusted-channel assumption in classical schemes~\cite{bonawitz2017} is not merely a convenience but a necessity: no classical protocol achieves information-theoretic privacy when every channel is observed. Quantumly the barrier is lifted: \Cref{prop:vga-robust} gives per-coordinate $\delta$-security conditional on verification, and composing this additively across coordinates and rounds (\Cref{lem:chain}) yields, for any $\varepsilon>0$, composable $\varepsilon$-secure aggregation at sufficiently many verified GHZ copies, at a $2\times$ overhead in GHZ copies for verification and with no trusted-channel assumption.

\subsection{Communication complexity}
The lower bounds in \Cref{sec:tieaudit} draw on two results from one-way communication complexity. The Gap-Hamming Distance lower bound of Chakrabarti and Regev~\cite{chakrabartiregev} gives
the tight $\widetilde\Omega(\tau^{-2})$ classical lower bound for \textsc{GapIP}; the quantum lower bound uses Razborov's symmetric-predicate technique~\cite{razborov}. The exponential separation for \textsc{TieAudit} reduces to the Boolean Hidden Matching lower bound of Gavinsky et al.~\cite{gavinsky2008exponential}, which established the first exponential separation between quantum and classical one-way communication complexity; Bar-Yossef, Jayram, and Kerenidis~\cite{baryossef2004exponential} give an alternative
construction. Our results instantiate these separations in concrete gradient statistics rather than abstractly defined communication problems.

\section{The quantum ring all-reduce primitive}\label{sec:protocol}

The goal of the quantum ring all-reduce is to compute $\bar{g}=\frac{1}{N}\sum_j g_j$ with minimum qubit communication while preserving the round structure of the classical primitive. The protocol has two phases: a \emph{setup} phase in which entanglement is distributed before any gradient exists, and an \emph{online} phase in which each training step runs the all-reduce using that entanglement.

\subsection{Setup: pre-sharing entanglement}

Before training begins, every pair of adjacent nodes $(j,j+1)$ pre-shares $\lceil Pb/2\rceil$ EPR pairs $\ket{\Phi^+}=\tfrac{1}{\sqrt{2}}(\ket{00}+\ket{11})$, one pair per two classical bits of gradient that will cross that link. This resource is \emph{input-independent}: it can be generated, certified, and distributed before any training data or model gradient is seen, which is both a practical convenience and the key property the security analysis exploits (see \Cref{sec:security}).

\subsection{Online: superdense-coding all-reduce}

The bandwidth advantage of replacing classical messages with qubits is fundamentally capped at $2\times$ by the entanglement-assisted Holevo capacity; the protocol below meets this bound exactly.

The classical ring all-reduce proceeds in two sub-phases: a \emph{scatter-reduce}, in which each node accumulates one $P/N$-dimensional partial sum by passing it clockwise around the ring, followed by an \emph{all-gather}, in which the completed mean is
disseminated counter-clockwise. Each sub-phase requires $N-1$ rounds, and $Pb$ classical bits cross every cut of the ring per sub-phase ($2Pb$ in total).

A straightforward quantum version would replace each classical message of $m$ bits with $\lceil m/2\rceil$ qubits. Node $j$ encodes two classical bits into one qubit by applying one of the four Pauli operations $\{I,X,iY,Z\}$ to its half of a pre-shared EPR pair;
the receiving node recovers both bits via a Bell measurement (superdense coding~\cite{bw92}). The round structure, the all-reduce schedule, and the gradient recovery arithmetic are identical to the classical case; only the physical carriers change. The online qubit cost across every link is therefore $\frac{Pb}{2}(1+o(1))$, a factor of two reduction.

\begin{remark}[Resource accounting across rounds]
Superdense coding consumes entanglement: the Bell measurement that decodes each message destroys the EPR pair it uses, so the pairs are not reusable and a fresh supply must be distributed for every training round.  Because entanglement cannot be created across a cut by local operations and classical communication, distributing one EPR pair across a link costs one qubit channel-use across that link, and each online qubit consumes exactly one such pair.  Counting both phases, the total number of carriers crossing every link therefore equals the classical baseline: consistent with
the Holevo bound, superdense coding does not reduce the \emph{total} communication, it halves only the \emph{online} portion.  The factor of two is thus a relocation of communication off the critical path rather than a net reduction in lifetime channel use, and this is the operationally relevant quantity.  The setup phase is input-independent and can be scheduled during computation or on a dedicated entanglement-distribution layer, whereas the online phase runs inside the all-reduce barrier that blocks all $N$ workers.  In the communication-bound regime typical of large-scale distributed training, where barrier bandwidth on the critical path is the scarce resource, halving the online volume is the saving that matters even though total channel use is unchanged; an ideal setup for our architecture.
\end{remark}

The factor is tight: a cut separating any node set $A$ from $B$ must convey the Shannon entropy of $\sum_{j\in B}g_j$ equals $Pb-O(\log N)$ bits (the sum of uniform $b$-bit words loses only carry overhead), and pre-shared entanglement is input-independent so by the
no-communication theorem it carries none of that entropy on its own. Each qubit crossing the cut conveys at most $C_E=2$ classical bits even with unlimited entanglement, the entanglement-assisted Holevo bound~\cite{bsst}. Hence at least $\tfrac{Pb}{2}-O(\log N)$ qubits must cross every cut in the online phase: the protocol meets the lower bound.

\begin{figure}[htbp]
\centering
\resizebox{\linewidth}{!}{%
\begin{tikzpicture}[
  cel/.style={draw=black!35, minimum width=1.9cm, minimum height=0.42cm,
              font=\scriptsize, align=center, inner sep=1pt},
  wlb/.style={cel, fill=gray!12, font=\scriptsize\bfseries},
  hdr/.style={minimum width=1.9cm, minimum height=0.38cm,
              font=\scriptsize\bfseries, text=black!55, align=center},
  sti/.style={font=\small\bfseries, align=center},
  sar/.style={-{Stealth[scale=0.9]}, thick, black!45},
]
\colorlet{Ca}{red!28}
\colorlet{Cb}{green!40}
\colorlet{Cc}{blue!22}
\colorlet{Cps}{orange!32}
\colorlet{Crr}{teal!30}

\node[sti] at (3.8, 0.45) {Step~0: Initial State};
\node[hdr] at (0.95, 0.00) {};
\node[hdr] at (2.85, 0.00) {Chunk 1};
\node[hdr] at (4.75, 0.00) {Chunk 2};
\node[hdr] at (6.65, 0.00) {Chunk 3};
\node[wlb] at (0.95,-0.42) {Worker A};
  \node[cel,fill=Ca] at (2.85,-0.42) {$A_1$};
  \node[cel,fill=Cb] at (4.75,-0.42) {$A_2$};
  \node[cel,fill=Cc] at (6.65,-0.42) {$A_3$};
\node[wlb] at (0.95,-0.84) {Worker B};
  \node[cel,fill=Ca] at (2.85,-0.84) {$B_1$};
  \node[cel,fill=Cb] at (4.75,-0.84) {$B_2$};
  \node[cel,fill=Cc] at (6.65,-0.84) {$B_3$};
\node[wlb] at (0.95,-1.26) {Worker C};
  \node[cel,fill=Ca] at (2.85,-1.26) {$C_1$};
  \node[cel,fill=Cb] at (4.75,-1.26) {$C_2$};
  \node[cel,fill=Cc] at (6.65,-1.26) {$C_3$};
\node[draw=purple!45, fill=purple!6, rounded corners=2pt, font=\scriptsize,
      text=purple!80!black, align=center, inner sep=5pt]
  at (10.5,-0.63) {Setup: pre-share $\ket{\Phi^+}$\\on each ring edge\\(0 training-time bits)};

\draw[sar] (3.8,-1.65) -- node[right,font=\scriptsize,text=black!48]{round 1} (3.8,-2.10);

\node[sti] at (3.8,-2.50) {Scatter-Reduce Round~1:
  \quad{\small(A$\!\to\!$B: $A_1$,\; B$\!\to\!$C: $B_2$,\; C$\!\to\!$A: $C_3$)}};
\node[hdr] at (0.95,-2.95) {};
\node[hdr] at (2.85,-2.95) {Chunk 1};
\node[hdr] at (4.75,-2.95) {Chunk 2};
\node[hdr] at (6.65,-2.95) {Chunk 3};
\node[wlb] at (0.95,-3.37) {Worker A};
  \node[cel,fill=Ca]  at (2.85,-3.37) {$A_1$};
  \node[cel,fill=Cb]  at (4.75,-3.37) {$A_2$};
  \node[cel,fill=Cps] at (6.65,-3.37) {$A_3{+}C_3$};
\node[wlb] at (0.95,-3.79) {Worker B};
  \node[cel,fill=Cps] at (2.85,-3.79) {$A_1{+}B_1$};
  \node[cel,fill=Cb]  at (4.75,-3.79) {$B_2$};
  \node[cel,fill=Cc]  at (6.65,-3.79) {$B_3$};
\node[wlb] at (0.95,-4.21) {Worker C};
  \node[cel,fill=Ca]  at (2.85,-4.21) {$C_1$};
  \node[cel,fill=Cps] at (4.75,-4.21) {$B_2{+}C_2$};
  \node[cel,fill=Cc]  at (6.65,-4.21) {$C_3$};
\node[font=\scriptsize\bfseries,text=black!70] at (11.5,-2.95) {Classical \quad Quantum};
\draw[black!22] (8.3,-3.12) -- (13.3,-3.12);
\node[cel,minimum width=1.6cm,draw=none,align=right,font=\scriptsize]
  at (9.05,-3.54) {Per link,\\per chunk:};
\node[cel,minimum width=1.5cm,fill=gray!10] at (10.65,-3.54) {$2b$ bits};
\node[cel,minimum width=1.5cm,fill=blue!8,draw=blue!30] at (12.35,-3.54) {$b$ qubits};
\node[cel,minimum width=1.6cm,draw=none,align=right,font=\tiny,text=black!55]
  at (9.05,-3.96) {(setup):};
\node[cel,minimum width=1.5cm,fill=gray!5,font=\tiny,text=black!35] at (10.65,-3.96) {\textemdash};
\node[cel,minimum width=1.5cm,fill=purple!6,draw=purple!30,font=\tiny,text=purple!75!black]
  at (12.35,-3.96) {$1$ EPR pair};

\draw[sar] (3.8,-4.60) -- node[right,font=\scriptsize,text=black!48]{round 2} (3.8,-5.05);

\node[sti] at (3.8,-5.45) {Scatter-Reduce Round~2: each worker owns one fully-reduced chunk};
\node[hdr] at (0.95,-5.90) {};
\node[hdr] at (2.85,-5.90) {Chunk 1};
\node[hdr] at (4.75,-5.90) {Chunk 2};
\node[hdr] at (6.65,-5.90) {Chunk 3};
\node[wlb] at (0.95,-6.32) {Worker A};
  \node[cel,fill=Ca!30,text=black!38] at (2.85,-6.32) {$A_1$};
  \node[cel,fill=Crr]                 at (4.75,-6.32) {$\bar{g}_2$};
  \node[cel,fill=Cps!45,text=black!38] at (6.65,-6.32) {$A_3{+}C_3$};
\node[wlb] at (0.95,-6.74) {Worker B};
  \node[cel,fill=Cps!45,text=black!38] at (2.85,-6.74) {$A_1{+}B_1$};
  \node[cel,fill=Cb!30,text=black!38] at (4.75,-6.74) {$B_2$};
  \node[cel,fill=Crr]                 at (6.65,-6.74) {$\bar{g}_3$};
\node[wlb] at (0.95,-7.16) {Worker C};
  \node[cel,fill=Crr]                 at (2.85,-7.16) {$\bar{g}_1$};
  \node[cel,fill=Cps!45,text=black!38] at (4.75,-7.16) {$B_2{+}C_2$};
  \node[cel,fill=Cc!30,text=black!38] at (6.65,-7.16) {$C_3$};

\draw[sar] (3.8,-7.55) -- node[right,font=\scriptsize,text=black!48]{round 1} (3.8,-8.00);

\node[sti] at (3.8,-8.40) {All-Gather Round~1:
  \quad{\small(A$\!\to\!$B: $\bar{g}_2$,\; B$\!\to\!$C: $\bar{g}_3$,\; C$\!\to\!$A: $\bar{g}_1$)}};
\node[hdr] at (0.95,-8.85) {};
\node[hdr] at (2.85,-8.85) {Chunk 1};
\node[hdr] at (4.75,-8.85) {Chunk 2};
\node[hdr] at (6.65,-8.85) {Chunk 3};
\node[wlb] at (0.95,-9.27) {Worker A};
  \node[cel,fill=Crr]                at (2.85,-9.27) {$\bar{g}_1$};
  \node[cel,fill=Crr]                at (4.75,-9.27) {$\bar{g}_2$};
  \node[cel,draw=black!15,text=black!22] at (6.65,-9.27) {$\cdots$};
\node[wlb] at (0.95,-9.69) {Worker B};
  \node[cel,draw=black!15,text=black!22] at (2.85,-9.69) {$\cdots$};
  \node[cel,fill=Crr]                at (4.75,-9.69) {$\bar{g}_2$};
  \node[cel,fill=Crr]                at (6.65,-9.69) {$\bar{g}_3$};
\node[wlb] at (0.95,-10.11) {Worker C};
  \node[cel,fill=Crr]                at (2.85,-10.11) {$\bar{g}_1$};
  \node[cel,draw=black!15,text=black!22] at (4.75,-10.11) {$\cdots$};
  \node[cel,fill=Crr]                at (6.65,-10.11) {$\bar{g}_3$};
\node[font=\scriptsize\bfseries,text=black!70] at (11.5,-8.85) {Classical \quad Quantum};
\draw[black!22] (8.3,-9.03) -- (13.3,-9.03);
\node[cel,minimum width=1.6cm,draw=none,align=right,font=\scriptsize]
  at (8.9,-9.27) {Scatter-reduce:};
\node[cel,minimum width=1.5cm,fill=gray!8] at (10.65,-9.27) {$Pb$ bits};
\node[cel,minimum width=1.5cm,fill=blue!7,draw=blue!25] at (12.35,-9.27) {$\tfrac{Pb}{2}$ qubits};
\node[cel,minimum width=1.6cm,draw=none,align=right,font=\scriptsize]
  at (9.05,-9.69) {All-gather:};
\node[cel,minimum width=1.5cm,fill=gray!8] at (10.65,-9.69) {$Pb$ bits};
\node[cel,minimum width=1.5cm,fill=blue!7,draw=blue!25] at (12.35,-9.69) {$\tfrac{Pb}{2}$ qubits};
\draw[black!22] (8.3,-9.90) -- (13.3,-9.90);
\node[cel,minimum width=1.6cm,draw=none,align=right,font=\scriptsize\bfseries]
  at (9.05,-10.11) {Total:};
\node[cel,minimum width=1.5cm,fill=gray!15,font=\scriptsize\bfseries]
  at (10.65,-10.11) {$2Pb$ bits};
\node[cel,minimum width=1.5cm,fill=blue!14,draw=blue!40,font=\scriptsize\bfseries]
  at (12.35,-10.11) {$Pb$ qubits};

\draw[sar] (3.8,-10.50) -- node[right,font=\scriptsize,text=black!48]{round 2} (3.8,-10.95);

\node[sti] at (3.8,-11.35) {All-Gather Round~2 (Final): all workers hold all reduced chunks};
\node[hdr] at (0.95,-11.80) {};
\node[hdr,text=teal!65!black] at (2.85,-11.80) {$\bar{g}_1$};
\node[hdr,text=teal!65!black] at (4.75,-11.80) {$\bar{g}_2$};
\node[hdr,text=teal!65!black] at (6.65,-11.80) {$\bar{g}_3$};
\node[wlb] at (0.95,-12.22) {Worker A};
  \node[cel,fill=Crr] at (2.85,-12.22) {$\bar{g}_1$};
  \node[cel,fill=Crr] at (4.75,-12.22) {$\bar{g}_2$};
  \node[cel,fill=Crr] at (6.65,-12.22) {$\bar{g}_3$};
\node[wlb] at (0.95,-12.64) {Worker B};
  \node[cel,fill=Crr] at (2.85,-12.64) {$\bar{g}_1$};
  \node[cel,fill=Crr] at (4.75,-12.64) {$\bar{g}_2$};
  \node[cel,fill=Crr] at (6.65,-12.64) {$\bar{g}_3$};
\node[wlb] at (0.95,-13.06) {Worker C};
  \node[cel,fill=Crr] at (2.85,-13.06) {$\bar{g}_1$};
  \node[cel,fill=Crr] at (4.75,-13.06) {$\bar{g}_2$};
  \node[cel,fill=Crr] at (6.65,-13.06) {$\bar{g}_3$};

\end{tikzpicture}}%
\caption{Quantum ring all-reduce for $N=3$ workers, each holding a gradient
partitioned into 3 equal chunks.
\emph{Orange}: partial sum accumulated so far.\\
\emph{Teal}: fully reduced chunk $\bar{g}_k=\frac{1}{N}\sum_j g_{j,k}$.\\
\textbf{Scatter-reduce} ($N-1=2$ rounds, clockwise): each worker sends one
chunk per round; after round~2 every worker holds exactly one fully-reduced
chunk (one third of the gradient).\\
\textbf{All-gather} ($N-1=2$ rounds): each worker broadcasts its reduced chunk;
after round~2 every worker holds all chunks.\\
\emph{Right panels}: superdense coding~\cite{bw92} allows $2b$ gradient bits to
cross a ring link using only $b$ qubits, given a pre-shared EPR pair
$\ket{\Phi^+}$ loaded in setup; the online qubit cost is $Pb$ versus $2Pb$ bits
classically, a factor-of-two saving that meets the entanglement-assisted
Holevo lower bound.}
\label{fig:qring}
\end{figure}

Figure~\ref{fig:qring} shows this superdense-coding protocol worked through for $N=3$.  The construction so far is a direct application of superdense coding, included because it targets the defining cost of ring all-reduce, its per-link bandwidth; it is the simplest component of the paper, and its value is as a substrate rather than a result in its own right.  Section~\ref{subsec:GHZ} recasts the same all-reduce in terms of a shared GHZ state rather than EPR pairs, and it is this reformulation carries the conceptual
content of what follows: the optimal-probe characterisation of \Cref{alg:ghz}, the security guarantees of \Cref{sec:security}, and the structure exploited by the conflict-detection separations of \Cref{sec:tieaudit}.

\subsection{GHZ phase encoding: an equivalent formulation and its optimality} \label{subsec:GHZ}

An alternative formulation encodes gradients directly in the phases of a shared GHZ state rather than compressing classical bit strings. This view is the basis of the security protocol in \Cref{sec:security}, and it also yields the sharpest statement of
\emph{why} GHZ is the right resource.

\begin{definition}[Local phase encoding]\label{def:encoding}
Each node $j$ holds one qubit of a shared $N$-qubit pure state $\ket{\psi}$ and applies
$U_j(\theta_j)=e^{-i\theta_j Z_j/2}$, where $\theta_j$ is a $b$-bit encoding of its
gradient coordinate. Writing $H=\tfrac{1}{2}\sum_j w_j Z_j$ for weights $w_j>0$ (the equal-weight
case $w_j=1$ encodes a simple sum), the post-encoding state along the direction
$\theta_j\equiv\vartheta$ is $e^{-i\vartheta H}\ket{\psi}$.
Every node measures in the $X$ basis, broadcasts its outcome bit, and all nodes recover
the aggregate $\Phi=\sum_j\theta_j$ by parity.
\end{definition}

\begin{definition}[Quantum Fisher information~\cite{braunsteincaves}]
For the family $\ket{\psi_\vartheta}=e^{-i\vartheta H}\ket{\psi}$, the quantum Fisher
information is $F_Q=4\,\mathrm{Var}_\psi(H)$. Every locally unbiased estimator of
$\Phi$ from $\nu$ independent copies obeys
$\mathrm{Var}(\hat\Phi)\ge 1/(\nu F_Q)$.
$F_Q$ bounds the information extractable by \emph{any} measurement strategy, including
adaptive LOCC with classical broadcast.
\end{definition}

\begin{algorithm}[htbp]
\caption{Quantum all-reduce: GHZ phase encoding formulation}
\label{alg:ghz}
\begin{algorithmic}[1]
\Require{Each node $j\in\{1,\dots,N\}$ holds a gradient coordinate $\theta_j\in[0,2\pi)$ encoded to $b$-bit precision. 
\emph{AND} Each node $j$ holds a weight $w_j>0$.
\emph{AND} A fresh $N$-qubit GHZ state $\ket{\psi}=\tfrac{1}{\sqrt2}(\ket{0^N}+\ket{1^N})$ is supplied to the ring before each of the $b$ rounds below.}
\Ensure{$\hat\Phi$, an estimate of $\Phi=\sum_j w_j\theta_j \bmod 2\pi$ accurate to $b$ bits, identical at every node}
\For{$k = b-1, b-2, \dots, 0$} 
  \State Distribute one fresh copy of $\ket{\psi}$, one qubit per node
  \For{each node $j$ \textbf{in parallel}}
    \State $\theta_j^{(k)} \gets w_j\,\theta_j \cdot 2^{k}\ (\mathrm{mod}\ 2\pi)$
    \State Apply $U_j(\theta_j^{(k)}) = e^{-i\theta_j^{(k)} Z_j/2}$ to its half of $\ket{\psi}$
    \State Measure local qubit in the $X$ basis $\to$ outcome $m_j^{(k)}\in\{0,1\}$
    \State Broadcast $m_j^{(k)}$ to all other nodes \Comment{$N-1$ classical bits}
  \EndFor
  \State Every node computes the parity bit $p^{(k)} \gets \bigoplus_j m_j^{(k)}$
\EndFor
\State Every node combines $(p^{(b-1)},\dots,p^{(0)})$ into an estimate $\hat\Phi$ 
\Return $\hat\Phi$
\end{algorithmic}
\end{algorithm}

\begin{theorem}[GHZ uniquely maximises linear-aggregate precision]\label{thm:qfi}
Fix the resource model of \Cref{def:encoding}: a single $N$-qubit pure probe
state $\ket{\psi}$\footnote{That is, $\rho=\ket{\psi}\bra{\psi}$ is rank one, not a
statistical mixture. This restriction matters: $F_Q$ is convex in $\rho$, so any mixed
probe has $F_Q$ no larger than the best pure state in its decomposition, and the
GHZ-uniqueness claim would need restating in terms of $\rho$ if mixed probes were
allowed.}, encoding restricted to local phase rotations $U_j(\theta_j)$ (no
entangling gates during encoding, no ancillas beyond the $N$ encoding qubits), and
$\nu$ independent, identically prepared, unentangled copies of the protocol. For every
such $\ket{\psi}$ and weights $w_j>0$,
\[
F_Q = 4\operatorname{Var}_\psi(H) \le \Bigl(\sum_j w_j\Bigr)^2,
\]
with equality if and only if
$\ket{\psi}=\tfrac{1}{\sqrt{2}}(\ket{0^N}+e^{i\varphi}\ket{1^N})$, a GHZ-type state.
Any amplitude $\gamma$ on a non-extremal eigenspace of $H$ with eigenvalue $\lambda_\chi$ incurs the quantitative penalty
$F_Q = \Bigl(\sum_j w_j\Bigr)^2\Bigl(1 - |\gamma|^2\bigl(1 - \lambda_\chi^2/\lambda_{\max}^2\bigr)\Bigr) - 4|\gamma|^4\lambda_\chi^2$.
\end{theorem}

\begin{proof}
The spectrum of $H=\tfrac{1}{2}\sum_j w_j Z_j$ lies in
$[\lambda_{\min},\lambda_{\max}]=[-\tfrac{1}{2}\sum w_j,\tfrac{1}{2}\sum w_j]$.
Popoviciu's variance inequality gives
$\mathrm{Var}_\psi(H)\le\bigl(\tfrac{\lambda_{\max}-\lambda_{\min}}{2}\bigr)^2$
with equality iff the spectral distribution of $\ket{\psi}$ puts mass $\tfrac12$ on each
extremum. Since all $w_j>0$, the extreme eigenspaces are one-dimensional ($\ket{0^N}$ and
$\ket{1^N}$), forcing $\ket{\psi}$ to have the GHZ form up to a relative phase. The
penalty formula follows by direct variance computation on the three-component
decomposition.
\end{proof}

Reaching aggregate precision $\tau$ costs $\nu\ge 1/(\tau^2 F_Q)$ encoded copies, so GHZ-type states are the unique optimal probe
shot-for-shot among all $N$-qubit states for any linear aggregate $\sum_j w_j\theta_j$. This optimality is specific to linear functionals; for nonlinear functionals of the gradients, no single probe state is universally optimal and performance becomes instance- and structure-dependent.

\begin{remark}[Saturation of the Quantum Fisher Information bound by the parity protocol]
\Cref{thm:qfi} bounds $F_Q$, which by the quantum Cram\'er--Rao bound of
\Cref{def:encoding} is itself only a lower bound on the variance of any estimator;
attaining it in general requires the
optimal (symmetric-logarithmic-derivative) measurement, not the concrete $X$-basis,
parity-readout scheme of \Cref{def:encoding}.  For the GHZ state these coincide: the
SLD measurement for $H=\tfrac12\sum_j w_jZ_j$ on
$\tfrac{1}{\sqrt2}(\ket{0^N}+e^{i\varphi}\ket{1^N})$ is exactly the transversal $X$-basis
measurement with parity readout, so the protocol of \Cref{def:encoding} saturates the
bound of \Cref{thm:qfi} with no further assumptions. This is what upgrades the
theorem's claim from "GHZ uniquely maximises $F_Q$" to "GHZ is the unique optimal probe
\emph{for the stated protocol}, shot-for-shot."
\end{remark}

\section{A secure quantum ring all-reduce protocol}\label{sec:security}

The quantum communication advantage of \Cref{sec:protocol} uses entanglement as a resource established before any gradient exists.  
That entanglement pre-sharing enables a \emph{information-theoretic privacy} guarantee, that cannot be achieved by any classical protocol in a setting where an adversary observes every channel, including the setup phase. Building on the GHZ phase encoding formulation of \Cref{subsec:GHZ}, we propose the Verified GHZ Aggregation (VGA) protocol, which achieves \emph{information-theoretic private} aggregation at a cost of twice as many GHZ copies as the unverified phase encoding protocol.

\subsection{Model and definitions}

Eve is read-only on the authenticated classical channels (she sees every classical symbol but cannot modify it) and may act arbitrarily on the quantum systems in transit, including tampering; such tampering is caught, except with small probability, by the verification step of \Cref{sec:security}. 
By no-cloning~\cite{wootters1982single} she cannot copy a quantum system in transit without disturbing it, so any such attempt is caught by verification rather than yielding a silent copy. Nodes hold independent private randomness but begin with no pre-shared secrets. Leakage about node $j$'s gradient $g_j$ is measured as the mutual information $I(g_j;\,V_E)$, where $V_E$ is Eve's total accumulated view. Eve's goal is to learn as much as possible about individual nodes' gradients beyond what the aggregate output $\bar{g}$ already discloses. In distributed learning this is a concrete threat: gradient values have been shown to allow reconstruction of private training data and membership inference~\cite{zhang2022federated}; our metric $I(g_j;\,V_E)$ measures exactly how much Eve's view of the protocol exceeds what the aggregate alone reveals.

\begin{definition}[Wiretapped-setup model]\label{def:untrusted}
This model runs the same $N$-node ring topology and scatter-reduce / all-gather structure as \Cref{sec:protocol}, but with no prior correlations between nodes.  Every channel (classical and quantum) in both the setup phase (quantum state distribution) and the online phase is exposed to an adversary Eve, who is read-only on the authenticated classical channels and may act arbitrarily (including tampering) on the quantum systems in transit.  Classical channels are authenticated.  Nodes hold independent private randomness.  Leakage is $I(g_j;\,V_E)$, where $V_E$ is Eve's total view.
\end{definition}

To state security precisely we use composable simulation security: a protocol is \emph{secure} if whatever Eve learns from participating, she could equally have learned by interacting with a simulator that sees only the aggregate output $\bar{g}$ and nothing else, meaning the real protocol leaks no information that the output alone would not already reveal.

\begin{definition}[$\varepsilon$-secure aggregation (composable)~\cite{mullerquaderenner}]
\label{def:composable}
A protocol is $\varepsilon$-secure if for every adversary there exists a simulator
interacting only with an ideal functionality, which collects inputs, returns
$\bar g$ (or aborts), and reveals nothing else, such that the joint state of honest
outputs and the adversary's view in the real protocol is within trace distance
$\varepsilon$ of the corresponding state in the simulated ideal execution.
Trace-distance security composes additively.
\end{definition}

The trace-distance bound $\varepsilon$ upper-bounds the distinguishing advantage of any measurement between the real and ideal experiments.

\subsection{Classical impossibility}

\begin{theorem}[No classical IT privacy under wiretapped setup]
\label{thm:cimposs}
Any classical protocol in \Cref{def:untrusted} that computes $\bar g$ with error
probability $\epsilon$ admits an input distribution under which
\[
I(g_j;\,T_{\mathrm{full}}) \;\ge\; (1-\epsilon)\,H(g_j) - h(\epsilon),
\]
where $H(\cdot)$ denotes Shannon entropy, $h(\cdot)$ binary entropy, and
$T_{\mathrm{full}}$ is the full protocol transcript.
Vanishing leakage is therefore unachievable at any communication cost.
\end{theorem}

\begin{proof}
We reduce to secret-key agreement.  Set $g_j = K$ uniformly random and all other
nodes run honestly on input $0$, so the honest output is $\bar g = K/N$, from which
any other node $k$ recovers $K$ exactly.  By Fano's inequality,
$H(K\mid T_{\mathrm{full}}, R_k) \le h(\epsilon) + \epsilon\,H(K)$, where $R_k$ is
node $k$'s private randomness.  If the transcript satisfied
$I(K;\,T_{\mathrm{full}})\le\ell$, then nodes $j$ and $k$ (starting from
independent randomness) could agree, via authenticated public discussion over the
already-transmitted transcript, on a key of entropy at least
$H(K)-\ell-h(\epsilon)-\epsilon\,H(K)$ hidden from any passive observer. The secret-key capacity of independent sources over public authenticated channels is zero~\cite{maurer1993,ahlswede1993}; hence $\ell \ge (1-\epsilon)H(K)-h(\epsilon)$.
\end{proof} 

This closes, in the strongest form, the ``pre-shared one-time pad'' objection:
classical masking schemes~\cite{bonawitz2017} achieve information-theoretically private
aggregation if pads can be delivered secretly, and \Cref{thm:cimposs} shows that
a fundamental information-theoretic barrier when every channel is observed.

\subsection{Verified GHZ aggregation} 

The VGA protocol builds on the phase encoding formulation of \Cref{def:encoding}. The challenge of this setup is that Eve can observe the GHZ distribution step  and may tamper with the states en route. The VGA protocol handles this attack by executing a random stabilizer test before any gradient is encoded: surviving copies are certified close to genuine GHZ states, and because certification happens before inputs arrive, abort decisions are input-independent, mirroring QKD protocols, where a random subset of transmitted systems is sacrificed for parameter estimation before any key is extracted.

\begin{tcolorbox}[colback=teal!6, colframe=teal!50,
  title=\textbf{Protocol VGA (Verified GHZ Aggregation), per gradient coordinate}]
\textbf{Step~1 (Distribution)}
A designated node prepares $2t$ copies of
$\ket{\mathrm{GHZ}_N}=\tfrac{1}{\sqrt{2}}(\ket{0^N}+\ket{1^N})$ and routes one qubit
of each copy to every other node around the ring.

\textbf{Step~2 (Verification)}
A public random string selects $t$ copies and, for each, announces a uniformly random
nontrivial stabilizer $g\in\mathcal{S}\setminus\{I\}$ of the GHZ group $\mathcal{S}$.
Every such $g$ is a tensor product of single-qubit Paulis, so all nodes measure
locally and broadcast outcomes. \textit{Abort} unless every tested $g$ yields $+1$.

\textbf{Step~3 (Aggregation)}
On the surviving $t$ copies, run phase encoding aggregation (\Cref{def:encoding}):
node $j$ applies $e^{-i\theta_j Z_j/2}$ to its share of each copy, all nodes
measure in the $X$ basis and broadcast the outcome bit, and the aggregate
$\Phi=\sum_j\theta_j$ is recovered by parity.
\end{tcolorbox}

We establish security in two layers. On exact GHZ states, every classical bit broadcast during readout is simulatable from the aggregate alone, so the protocol leaks nothing beyond its output (\Cref{lem:sim,lem:vga-ideal}). Verification then certifies that the accepted states are close to exact GHZ, and security degrades only gracefully in that distance (\Cref{prop:vga-robust}); the closeness itself is bounded by the sampling analysis that follows.

\begin{lemma}[Per-copy privacy, including the broadcast readout]\label{lem:sim}
For exact GHZ states, the joint view of any coalition 
$S\subsetneq[N]$ (their qubits, any applied deviations, Eve's residual system, and \emph{all classical bits
broadcast during the readout}) is simulatable from the coalition's own inputs
and the honest-sum $\Phi_H = \sum_{j\notin S}\theta_j$.  
Individual gradients $\theta_j$ for $j\notin S$ are never revealed beyond what $\Phi_H$ implies.
\end{lemma}

\begin{proof}
After honest encoding the shared state is
$\tfrac{1}{\sqrt{2}}\bigl(\ket{0^N}+e^{i\Phi_H}e^{i\phi_S}\ket{1^N}\bigr)$,
where $\phi_S$ is the phase applied by the coalition (known to the simulator).
The distribution of the honest nodes' $X$-measurement outcomes, conditioned on any
coalition measurement with outcome $\omega$, depends on the state only through the
relative phase $\Phi_H+\phi_S$:
\[
\Pr\!\Bigl[\bigoplus_{j\notin S} b_j = 0 \;\Big|\; \text{coalition outcome } \omega\Bigr]
= \tfrac{1}{2}\bigl(1 + \Re\,e^{i(\Phi_H+\phi_S)}c_\omega\bigr),
\]
where $c_\omega$ depends only on the coalition's (known) measurement.  The simulator
samples honest outcome bits from this distribution; individual $\theta_j$ never appear
beyond their contribution to $\Phi_H$.
\end{proof}

\begin{lemma}[Ideal-GHZ security, exact]\label{lem:vga-ideal}
On exact GHZ copies, VGA realises the ideal sums-only functionality of \Cref{def:composable} exactly: there is a simulator interacting only with the ideal
functionality (which returns $\Phi=\sum_j\theta_j$ or aborts) whose output is
identical to the real execution. In particular $\varepsilon=0$.
\end{lemma}

\begin{proof}
Distribution and verification are input-independent, so the simulator runs them against the real Eve, reproducing the accepted state and the abort decision identically. Eve's residual view is the broadcast readout together with her side system; this is the case $S=\varnothing$ of \Cref{lem:sim} (a purely external, passive Eve controls no inputs), for which $\Phi_H=\sum_j\theta_j=\Phi$, so the readout is simulatable exactly from $\Phi$. Composing over the $t$ surviving copies gives an exact simulator; hence $\varepsilon=0$.
\end{proof}

\begin{proposition}[Conditional robust security of VGA]\label{prop:vga-robust}
Condition on acceptance, and suppose the joint state $\rho_{C\bar T E}$ of the
$t$ surviving copies (honest register $C$) and Eve's system $E$ satisfies
\[
  \bigl\lVert \rho_{C\bar T E} - \tilde\rho_{C\bar T E}\bigr\rVert_1 \;\le\; \delta,
\]
where $\tilde\rho$ has honest register equal to $t$ exact GHZ copies
$\lvert\mathrm{GHZ}_N\rangle^{\otimes t}$ (Eve may hold arbitrary side
information). Then VGA is $\delta$-secure per coordinate in the sense of \Cref{def:composable}.
\end{proposition}

\begin{proof}
The remaining protocol (input encoding $U_j(\theta_j)$ on each honest share,
$X$-basis measurement, broadcast, and Eve's processing) is a fixed channel
$\Lambda_\theta$ depending on the inputs but not on the probe state. By
Lemma~\ref{lem:vga-ideal}, $\Lambda_\theta(\tilde\rho)$ equals the ideal
simulator output $\mathsf{S}(\Phi)$. Since trace distance is non-increasing
under the common channel $\Lambda_\theta$,
\[
  \bigl\lVert \Lambda_\theta(\rho)-\mathsf{S}(\Phi)\bigr\rVert_1
  =\bigl\lVert \Lambda_\theta(\rho)-\Lambda_\theta(\tilde\rho)\bigr\rVert_1
  \le \bigl\lVert \rho-\tilde\rho\bigr\rVert_1 \le \delta .
\]
Abort is input-independent and reproduced by the simulator, and trace-distance
security composes additively, giving $\delta$-security per coordinate.
\end{proof}

The verification of Step~2 is a Bouman--Fehr sampling strategy~\cite{BoumanFehr}
in the GHZ-stabiliser eigenbasis; a uniformly random nontrivial stabiliser flags
a corrupted copy ($\vec\sigma\neq0$) with probability $\ge\tfrac12$. Read at the
level of the purification, acceptance certifies (except with probability $\varepsilon_{\mathrm{samp}}$) that the surviving ensemble is within trace distance $\delta=O\!\bigl(\sqrt{\varepsilon_{\mathrm{samp}}}+\sqrt{\beta}\bigr)$
of an ideal-GHZ register, with $\beta=O\!\bigl(\log(1/\varepsilon_{\mathrm{samp}})/t\bigr)$.
Substituting this $\delta$ into Proposition~\ref{prop:vga-robust} yields the
security claim, with overhead a factor $2$ in GHZ copies and $O(1)$ broadcast
bits per tested copy.

We emphasise that this certification step is asymptotic. Three refinements remain for deployment-grade constants: 
\begin{enumerate}
    \item The $\tfrac12$ detection probability requires
an imperfect-detection Bouman--Fehr bound, or stabiliser amplification, to fix the
constant in $\beta$; 
    \item The entropic-uncertainty/decoupling conversion
\cite{TomamichelRenner,Renner} is used here in its asymptotic, leading-order form; the explicit finite-key correction terms, which would convert it into a security parameter at finite copy count $t$, are left to future work;
    \item A noise-tolerant acceptance threshold replacing the
perfect-pass rule remains to be established.
\end{enumerate}

None of these affects \Cref{prop:vga-robust} or the classical impossibility of
\Cref{thm:cimposs}; they bound only the achievable $\delta$.

\subsection{Multi-round privacy}

In practice, trainings with ring all-reduce are often completed in $T$ rounds. After observing $T$ consecutive aggregate sums, an adversary might try to piece together information about individual node gradients that no single sum would reveal alone. The theorem below shows our protocol is safe against this type of attack:  the coalition's full multi-round view is provably no more informative than the sequence of sums itself, VGA's per-round privacy extends to $T$ rounds at no extra cost.

\begin{theorem}[View reduction across rounds]\label{thm:view}
For $T$ rounds of ideal VGA (exact GHZ each round) against a static coalition $S\subsetneq[N]$, there exists a simulator that, given only the coalition's own inputs and the sequence of honest sums $\{\Phi_H^{(r)}\}_{r=1}^T$, reproduces the coalition's joint view (qubits, applied deviations, and all broadcast bits across all rounds) exactly. Consequently,
\[
I\!\bigl(\{g_j^{(r)}\}_{j\notin S,\,r\le T};\;\mathrm{view}_S\bigr)
\;=\;
I\!\bigl(\{g_j^{(r)}\}_{j\notin S,\,r\le T};\;\{\Phi_H^{(r)}\}_{r\le T}\bigr).
\]
Multi-round leakage collapses, with no loss, to the information in the sequence of
aggregates.
\end{theorem}

\begin{proof}
We proceed by induction on $T$.

\emph{Base case ($T=1$):}  This is \Cref{lem:sim}.

\emph{Inductive step:}  Assume the claim holds for $T-1$ rounds: there exists a
simulator $\mathrm{Sim}_{T-1}$ that reconstructs the coalition's view for rounds
$1,\ldots,T-1$ from the coalition's inputs and $\{\Phi_H^{(r)}\}_{r<T}$ alone.
For round $T$, the only cross-round information is \emph{public}: the model iterate
$w^{(T-1)}$ is a deterministic function of past broadcasts, already present in
$\mathrm{Sim}_{T-1}$'s output.
Conditioned on this shared public state, round $T$ is a fresh independent run of VGA
encoding $\Phi_H^{(T)}$; \Cref{lem:sim} supplies a simulator $\mathrm{Sim}_1^{(T)}$
for this single round from $\Phi_H^{(T)}$ alone.
Composing $\mathrm{Sim}_{T-1}$ with $\mathrm{Sim}_1^{(T)}$ gives the joint simulator
for all $T$ rounds.
The mutual information identity follows: exact simulation means the real and ideal joint
distributions coincide, so by the data-processing inequality
$I(\text{hidden gradients};\,\mathrm{view}_S)=I(\text{hidden gradients};\,\{\Phi_H^{(r)}\})$.
\end{proof}

\Cref{thm:view} reduces the question of ``What does running VGA for $T$ rounds reveal?'' 
to the question ``What does the sequence of sums reveal?''. 
For Gaussian inputs with a static personal component
$\mu_j$ contributing through the \emph{same} linear functional $M=\sum_j\mu_j$ every
round, data processing gives $I(\mu_1;\,\{\Phi_H^{(r)}\}_{r\le T})\le\tfrac12\ln\tfrac{N}{N-1}$,
a bound that is independent of $T$.  In regimes where gradients genuinely change across
rounds, the sums-only reduction makes Gaussian differential privacy efficient: nodes clip
and add local noise inside their phase encoding, the $N$ independent noise contributions
pool in the aggregate, and the released sum has $\sqrt{N}$ times less per-node noise
than local differential privacy (DP)~\cite{kairouz2021,dwork2006calibrating} (a framework
that bounds how much any single node's data can shift the output distribution, with
\emph{local} DP meaning each node adds calibrated noise before any sharing).

The following lemma extends these guarantees from the ideal oracle to the real VGA
protocol against a fully adaptive adversary.

\begin{lemma}[Adaptive chaining]\label{lem:chain}
Let $\Pi$ be $T$-round real VGA and $\Pi^\star$ the ideal sums-only oracle.  If each
round of $\Pi$, against any adversary adaptive on the public history, is
$\varepsilon_r$-close in trace distance to the corresponding ideal round, then $\Pi$
and $\Pi^\star$ are $\bigl(\sum_{r=1}^T\varepsilon_r\bigr)$-close in trace distance,
even against a fully adaptive cross-round adversary.
\end{lemma}

\begin{proof}
Define a sequence of \emph{hybrid} experiments
(a hybrid argument interpolates between two protocols one step at a time, bounding their
distance via the triangle inequality over consecutive swaps)
$H_0=\Pi, \dots, H_T=\Pi^\star$, where $H_m$ runs the ideal round for
$r\le m$ and the real round for $r>m$.  Adjacent hybrids differ only in round $m{+}1$;
conditioned on the identical public history up to round $m$, they are the real vs.\
ideal single round and therefore $\varepsilon_{m+1}$-close.  Trace distance is
non-increasing under the common CPTP continuation (remaining rounds plus Eve's
processing), so $\|H_m - H_{m+1}\|_1 \le \varepsilon_{m+1}$.  The triangle inequality
over $m=0,\dots,T{-}1$ gives $\|\Pi - \Pi^\star\|_1 \le \sum_r\varepsilon_r$.
Adaptivity is handled because the hybrid swap in round $m{+}1$ is conditioned on the
shared public history, generated identically by both $H_m$ and $H_{m+1}$ for all earlier
rounds.
\end{proof}

Together, \Cref{thm:view,lem:chain} show that $T$-round VGA leaks no more than the sequence of gradient sums it releases: the ideal multi-round view collapses exactly to the sums (\Cref{thm:view}), and real VGA is within $\sum_r\varepsilon_r$ in trace distance of that ideal (\Cref{lem:chain}), with $\varepsilon_r$ the per-round bound supplied by \Cref{prop:vga-robust}.

\section{Quantum advantage in gradient conflict detection under bandwidth constraints}\label{sec:tieaudit}

In distributed training, the ring all-reduce is not the end of the communication story. Once workers have aggregated their gradients, a server or parameter node typically distills the result into a compact signal and broadcasts it back to clients, each of which must interpret that signal against its own private model structure. The bandwidth available for this downstream broadcast can often be far smaller than what the ring itself uses: clients may be edge devices or be connected over a low-rate link. How much useful information a bandwidth-limited message can carry, and whether quantum communication offers an advantage, is the question we will now explore.

Unlike the homogeneous data-parallel setting of \Cref{sec:protocol,sec:security}, where every worker's gradient is a noisy estimate of the \emph{same} objective and averaging is precisely what one wants (it reduces variance), the conflict-detection regime arises when gradients optimise \emph{different} objectives or come from \emph{different} data distributions, so collapsing them by averaging can cancel signal rather than reduce noise, which is what makes detecting conflict, rather than blindly averaging, worthwhile. In such large-scale distributed training with heterogeneous or adversarial data distributions, \emph{gradient conflicts} are a central obstacle \cite{yu2020gradient,karimireddy2020scaffold}: two workers' gradient vectors may point in nearly opposite directions, and naively averaging them cancels signal or slows convergence. Detecting such conflicts after the all-reduce allows practitioners to selectively re-weight or drop conflicting updates, an idea that underlies conflict-aware federated learning algorithms and untying/refactoring schedules in over-parameterised networks.

The baseline for conflict detection is to broadcast all gradients: $O(Pb)$ bits per worker. The question is how much cheaper this can be made with quantum resources. This section characterises the quantum-classical separation for two natural variants of the problem, which arise from two different ways of formalising ``gradient conflict'':

\begin{itemize}
\item \emph{Margin conflict detection (\textsc{GapIP}$_\tau$):} two workers, each holding a unit-normalised gradient, test whether the signed inner product of their gradients falls above or below a threshold $\tau$. The quantum advantage is \textbf{quadratic}, tight on both sides.
\item \emph{Relational sign audit (\textsc{TieAudit}$_\epsilon$):} a server holds the global sign-gradient produced by signSGD; each client holds a private matching of parameter pairs (tied weights) and wants the fraction of its ties on which the global update is sign-inconsistent. The quantum advantage is \textbf{exponential} in $\log P$, arising from the hidden-matching structure of the relational task.
\end{itemize}

\subsection{Problem definitions}

Throughout this section, $P\in\mathbb{Z}_{>0}$ is the number of model parameters, $\tau\in(0,1)$ a margin threshold, and $\epsilon\in(0,\tfrac{1}{2})$ an additive estimation error. Communication complexity is one-way (client to server) unless stated otherwise.
We work with symmetric resources: both quantum and classical protocols may use
input-independent pre-shared correlations (EPR pairs or shared randomness respectively).
Readers new to one-way communication complexity and shared-resource models may consult
\Cref{app:cc} for a self-contained primer; \Cref{app:bhm} covers the Boolean Hidden
Matching problem on which \Cref{thm:tie} rests.

\begin{definition}[Margin conflict detection: $\textsc{GapIP}_\tau$]\label{def:gapip}
Two workers hold unit-norm gradients $g_j,g_k\in\mathbb{R}^P$, $\|g_j\|=\|g_k\|=1$.
Given the promise that $|\langle g_j,g_k\rangle|\ge\tau$, decide whether
$\langle g_j,g_k\rangle\ge\tau$ (\emph{aligned}) or $\langle g_j,g_k\rangle\le-\tau$
(\emph{conflicted}), with error probability at most $\tfrac{1}{3}$.
Communication is from one worker to the other (one-way) or interactive.
\end{definition}

\begin{definition}[Relational sign statistics: $\textsc{TieAudit}_\epsilon$]\label{def:tieaudit}
A server holds the global sign-gradient $s\in\{\pm1\}^P$, where $s_c=\mathrm{sign}([\bar{g}]_c)$
is the coordinate-wise sign of the all-reduce output from signSGD~\cite{bernstein2019signsgd}.
Each client $k$ holds a private perfect matching $M_k$ on the coordinate set $[P]=\{1,\dots,P\}$:
a set of $P/2$ disjoint index pairs $(i,j)$ induced by parameter ties (e.g. shared weights in a
convolutional layer, factor pairs in a matrix factorisation, or quantisation bins).
A pair $(i,j)\in M_k$ is \emph{sign-inconsistent} if $s_i\ne s_j$: the global update
disagrees on the tied coordinates.
The task is to estimate the fraction of conflicted pairs
\[
  f_k \;=\; \frac{|\{(i,j)\in M_k : s_i\ne s_j\}|}{|M_k|}
\]
to additive $\pm\epsilon$, given a one-way message from server to client.
This is an instance of the \emph{Boolean Hidden Matching}
problem~\cite{gavinsky2008exponential}: the server holds a bit string $x\in\{0,1\}^P$,
the receiver holds a perfect matching $M$ on $[P]$ and a label vector $w\in\{0,1\}^{P/2}$
with the promise that the edge-parity vector $(x_i\oplus x_j)_{(i,j)\in M}$ equals $w$ or
$\bar{w}$; the task is to decide which.
Setting $x_c=\mathbf{1}[s_c=+1]$ identifies the two problems. 
A self-contained account of Boolean Hidden Matching is given in \Cref{app:bhm}.
\end{definition}

\subsection{\textsc{GapIP}}

The \textsc{GapIP} problem can be described as deciding in which
side of $\pi/2$ the angle between $g_j$ and $g_k$ falls on (to within margin $\tau$). 
These types of angle decisions are precisely the use-case for amplitude estimation: namely, encoding each gradient into a quantum state, applying a Grover reflection product, and reading off the resulting angle without reconstructing the vectors.
Resolving an angle to precision $\Theta(\tau)$ takes $O(1/\tau)$ reflections,
independent of $P$; only the cost of a single round-trip carries the $\log P$
dependence.

Classically the angle is also defined, but there is no classical mechanism
analogous to coherent phase estimation for resolving it: any estimator built from
samples or random projections is bound by the standard $1/\sqrt{n}$ concentration
rate, so precision $\tau$ forces $n=\Theta(\tau^{-2})$ samples, by
Johnson--Lindenstrauss~\cite{dasgupta2003elementary}. This is the same shot-noise-versus-Heisenberg
contrast as~\Cref{thm:qfi}: amplitude estimation reaches $\tau$ precision in
$O(1/\tau)$ rounds where classical sampling needs $O(1/\tau^2)$.

We formalize this separation below.

\begin{theorem}[Gap Inner Product separation]\label{thm:margin}
With symmetric resources:
\begin{enumerate}
\item \emph{Quantum upper bound:}
  $\textsc{GapIP}_\tau$ is solvable with $\widetilde{O}(\tau^{-1}\log P)$ qubits
  of interactive communication. See \ref{app:gapip} for the protocol implementation.
\item \emph{Classical bounds:}
  The randomised communication complexity satisfies \\$R(\textsc{GapIP}_\tau)=\widetilde\Theta(\min(\tau^{-2},P))$ bits.
\item \emph{Quantum lower bound:}
  $Q(\textsc{GapIP}_\tau)=\widetilde\Omega(\tau^{-1})$ qubits. The separation
  is quadratic and the quantum protocol of~(1) is optimal up to log factors.
\end{enumerate}
\end{theorem}

\begin{proof}
\begin{enumerate}
    \item \emph{Quantum upper bound:}
Augment each unit vector to remove the sign ambiguity:
set $\hat{u}=(1,g_j)/\sqrt{2}$ and $\hat{v}=(1,g_k)/\sqrt{2}$ so that
$\langle\hat{u},\hat{v}\rangle = \frac{1+\langle g_j,g_k\rangle}{2}\in[0,1]$
is a monotone re-scaling of the signed inner product.
The two-party Grover operator
$G\;=\;(2\ket{\hat u}\bra{\hat u}-I)(2\ket{\hat v}\bra{\hat v}-I)$
rotates by angle $\theta=\arccos(\langle\hat{u},\hat{v}\rangle)$;
recovering $\langle g_j,g_k\rangle$ to precision $\pm\tau$ is equivalent to estimating
$\theta$ to precision $\Theta(\tau)$.
Kitaev phase estimation on $G$~\cite{bhmt} achieves this with $O(1/\tau)$
applications of $G$; each application requires one round-trip of $O(\log P)$ qubits.
Total: $\widetilde{O}(\tau^{-1}\log P)$ qubits.

    \item \emph{Classical bounds:}
\begin{itemize}
    \item \emph{Upper:} draw a random matrix $R\in\mathbb{R}^{k\times P}$ with i.i.d.\
$\mathcal{N}(0,1/k)$ entries (shared randomness so all parties hold the same $R$).
One worker broadcasts $Rg_j$ (a $k$-dimensional sketch).
By the distributional Johnson-Lindenstrauss lemma, for $k=O(\tau^{-2}\log(1/\delta))$,
$\langle Rg_j,Rg_k\rangle$ approximates $\langle g_j,g_k\rangle$ to $\pm\tau/2$
with probability $1-\delta$. Total: $O(\tau^{-2})$ words, i.e.\ $\widetilde{O}(\tau^{-2})$ bits.
The trivial upper bound of $O(P)$ bits (send the whole vector) gives
$R(\textsc{GapIP}_\tau)=\widetilde{O}(\min(\tau^{-2},P))$.
    \item \emph{Lower:} Gap-Hamming Distance $(\mathrm{GHD}_n)$ on $n$ bits reduces to
$\textsc{GapIP}_\tau$ with $n=4/\tau^2$: encode Alice's string as
$u=(-1)^x/\sqrt{n}\in\mathbb{R}^n$, so that gap $\tau$ in inner product corresponds
exactly to the Hamming gap in GHD. By~\cite{chakrabartiregev},
$R(\mathrm{GHD}_n)=\Omega(n)$, giving $R(\textsc{GapIP}_\tau)=\Omega(\tau^{-2})$.
\end{itemize}
    \item \emph{Quantum lower bound:}
The same embedding into $\mathrm{GHD}_{n}$ with $n=4/\tau^2$ gives
$Q(\textsc{GapIP}_\tau)\ge Q(\mathrm{GHD}_{4/\tau^2})$.
The quantum lower bound $Q(\mathrm{GHD}_n)=\Omega(\sqrt{n})$ follows from a blow-up
reduction from set-intersection~\cite{razborov} (repeat each coordinate $\sqrt{n}$ times;
Razborov's symmetric-predicate bound applies at the threshold jump).
Hence $Q(\textsc{GapIP}_\tau)\ge\Omega(\sqrt{4/\tau^2})=\Omega(1/\tau)$,
matching the upper bound of~(1) to log factors.
\end{enumerate}

\end{proof}

\begin{remark}[Scope of the quadratic advantage]\label{rem:margin}
The quantum protocol wins when $P>1/\tau$: otherwise the classical $O(P)$-bit upper bound
matches the quantum $O(1/\tau)$ bound and there is no separation.
For $N$ workers, classical conflict detection broadcasts one sketch per worker
($\widetilde{O}(N/\tau^{2})$ total), while the quantum protocol is inherently pairwise
($\widetilde{O}(N^2/\tau)$ total rounds); quantum wins in total communication only when
$N\lesssim1/\tau$.
Finally, the Heisenberg scaling in~(1) requires interaction: any one-round protocol
estimating a state overlap to $\pm\tau$ is shot-noise limited at
$\widetilde\Omega(\tau^{-2})$ copies~\cite{bhmt}, so the quadratic communication
advantage and $O(1)$-round operation are mutually exclusive.
\end{remark}

\subsection{$\textsc{TieAudit}_\epsilon$}

The margin result shows quantum is faster for \emph{numerical} conflict scores, by a quadratic factor. The exponential separation below does not come from precision, but the relational nature of the task. The main distinction is between a \emph{value} and a \emph{relation}. A value task asks for a number (e.g. $\langle g_j,g_k\rangle$) that tolerates additive error and can be compressed by sketching or sampling. A relational task asks for one valid fact drawn from a family that the \emph{receiver} privately chooses; no fixed short message can serve all possible matchings, because the receiver's choice is made \emph{after} the sender commits. This is why $\textsc{TieAudit}$ resists classical compression while admitting a small quantum message: the sender's phase state encodes all $P$ signs simultaneously in its amplitudes, and \emph{any} matching the receiver holds can extract a pair's parity in a single measurement, because the wrong-parity amplitude is exactly zero.

\begin{theorem}[Tie-audit separation]\label{thm:tie}
$\textsc{TieAudit}_\epsilon$ has one-way communication complexity
\[
  \text{Classical: }\;\Omega(\sqrt{P}\,), \qquad \text{Quantum: }\;O(\epsilon^{-2}\log P).
\]
\end{theorem}

\begin{proof}
\begin{enumerate}
        \item \emph{Quantum upper bound:}
The server prepares and sends $O(\epsilon^{-2})$ independent copies of the $\lceil\log_2 P\rceil$-qubit
phase state
\[
  \ket{s} \;=\; \frac{1}{\sqrt{P}}\sum_{c=1}^{P} s_c\ket{c},
\]
a coherent superposition over parameter indices with sign-gradient amplitudes.
A client holding matching $M_k$ measures each copy in the \emph{matching basis}
$\big\{(\ket{i}+\ket{j})/\sqrt{2},\;(\ket{i}-\ket{j})/\sqrt{2}\big\}_{(i,j)\in M_k}$.
This measurement collapses to a uniformly random pair $(i,j)\in M_k$ with outcome
$+$ if $s_i=s_j$ and $-$ if $s_i\ne s_j$: the wrong-parity amplitude is exactly zero,
because $s_i\ne s_j$ implies $\bra{i}+\bra{j}$ and $s_i\ket{i}+s_j\ket{j}$ are orthogonal.
Averaging $O(\epsilon^{-2})$ outcome bits estimates $f_k$ to additive $\pm\epsilon$ by
Chernoff bounds. Total: $O(\epsilon^{-2}\log P)$ qubits.

        \item \emph{Classical lower bound:}
Any estimator of $f_k$ to additive $\pm\tfrac{1}{4}$ distinguishes the two promise branches
of Boolean Hidden Matching: the fractions $\mathrm{wt}(w)/(P/2)$ and
$1-\mathrm{wt}(w)/(P/2)$ differ by at least $\tfrac{1}{2}$ for balanced $w$,
so an additive-$\tfrac{1}{4}$ estimate identifies the branch.
By the classical lower bound for Boolean Hidden Matching~\cite{gavinsky2008exponential},
any classical one-way protocol for this task requires $\Omega(\sqrt{P})$ bits.
\end{enumerate}
\end{proof}

It is important to note that the exponential separation concerns communication required to estimate a relational statistic of the aggregated gradient; \textbf{it does not} by itself \textbf{imply} an exponential \textbf{improvement in training performance}.

\begin{remark}[Why relational tasks resist compression]
No classical sketch of $o(\sqrt{P})$ bits can serve all possible client matchings, because
the matching $M_k$ is chosen after the server's message is fixed.
A Johnson-Lindenstrauss sketch\footnote{A Johnson-Lindenstrauss (JL) sketch compresses a vector $s\in\mathbb{R}^P$ to $k=O(\varepsilon^{-2}\log(1/\delta))$ dimensions via a random
linear map $A\in\mathbb{R}^{k\times P}$, with entries drawn i.i.d.\ $\pm1/\sqrt{k}$ or Gaussian.  The JL lemma guarantees $|\langle As,Au\rangle-\langle s,u\rangle|\leq\varepsilon$ for any \emph{fixed} query vector $u$ with probability $1-\delta$.  The sketch is therefore
useful when the query direction is known at encoding time, but cannot answer queries whose structure (here, the client matching $M_k$) is revealed only after the message is sent.}preserves inner products and norms but not the parity of a specific coordinate pair under an unanticipated matching. The quantum state $\ket{s}$ succeeds precisely because it encodes all $P$ signs in its amplitude structure and lets the receiver's choice of measurement basis post-select onto the relevant pair without any prior knowledge of the matching.
\end{remark}

\begin{remark}[Practical scope]\label{rem:tieaudit}
\Cref{thm:tie} establishes the first exponential one-way quantum advantage for a statistic
of the aggregated gradient in server-to-client gradient communication.
Two things to consider alongside the result.
\begin{enumerate}
    \item Training benefit is open: the natural use is an adaptive untying schedule
that breaks ties whose conflict fraction remains high across rounds (e.g.\ splitting a
shared weight when the global update consistently disagrees with the tie).
Whether consuming the \textsc{TieAudit} statistic improves convergence is plausible but
must be demonstrated experimentally.
    \item Quantum states cannot be broadcast, so serving $N$ clients requires $N$
separate copies of $\ket{s}$: total cost $\Theta(N\epsilon^{-2}\log P)$ qubits versus one
classical $O(\sqrt{P})$-bit broadcast.
The per-client advantage is exponential; the system-wide quantum advantage holds only
when $N\ll\sqrt{P}/(\epsilon^{-2}\log P)$.
\end{enumerate}
\end{remark}

\section{Conclusion and outlook}\label{sec:conclusion}

We introduced a quantum version of ring all-reduce and characterised the advantages that quantum communication can provide for distributed learning. Our results carry several qualifications. Our first result is a bandwidth-optimal quantum ring all-reduce primitive that uses pre-shared entanglement and superdense coding to reduce per-link communication by exactly a factor of two while preserving the classical round structure. We proved that this factor is optimal, and no entanglement-assisted protocol can reduce the communication below the entanglement-assisted Holevo limit. In the second result, the bandwidth analysis assumes noiseless quantum channels and perfect EPR pair generation; realistic hardware introduces fidelity loss that erodes the factor-of-two gain and requires a separate noise analysis. The VGA security proof is asymptotic: as noted in \Cref{sec:security}, finite-key constants and a noise-tolerant verification threshold remain to be established.
The gradient conflict detection results are worst-case communication complexity separations; whether the quantum advantage materialises in practice depends on the overhead of preparing and measuring the required states, which remains non-trivial on current hardware. Finally, quantum states cannot be broadcast: serving many clients requires independent copies of the quantum message, limiting system-wide advantage to the regime identified in \Cref{rem:tieaudit}.

Several avenues remain open for future work:
\begin{itemize}
    \item On the learning side, it remains to determine which gradient statistics beyond those studied here admit substantial quantum communication advantages, and whether acting on \textsc{TieAudit} output to drive adaptive model corrections yields measurable convergence improvements. 
    \item Regarding physical realism, every advantage we establish here assumes high-fidelity entanglement and noiseless channels, so real hardware will require entanglement distillation, quantum repeater networks to distribute entanglement across the distances separating workers, and ultimately fault-tolerant entanglement distribution.
Part of this is already built into VGA: its verification is a fidelity test, so noise tolerance amounts to setting its acceptance threshold above the hardware noise floor rather than demanding perfect agreement.
    \item Topology is a another such question: we use the ring because it is the bandwidth-optimal collective for data-parallel training, but the GHZ phase-encoding formulation of
\Cref{subsec:GHZ} is largely topology-agnostic, as the aggregate $\sum_j\theta_j$ is a global property of the shared GHZ state rather than of any particular interconnect, and
it is natural to ask whether the same advantages extend to tree all-reduce, parameter-server architectures, and the hierarchical collectives that combine intra- and
inter-node reduction. The per-link superdense-coding saving transfers to any topology directly, whereas the security construction depends on how the GHZ state is distributed and verified. 
\end{itemize}
Settling these questions will require characterising the end-to-end overhead of the protocol in a realistic network architecture, including entanglement generation rates, qubit coherence requirements, and integration with collective communication libraries such as NCCL and Horovod, which will ultimately decide whether the bandwidth saving and the privacy guarantee remain net-positive once the cost of producing the states they consume is counted.

More broadly, our results suggest a hierarchy of quantum advantages for distributed learning. When the task is to communicate the value of a global aggregate, quantum communication is constrained by information-theoretic limits and yields at most constant-factor improvements. When the task is to certify the process by which the aggregate was generated, quantum resources enable privacy guarantees unavailable to classical protocols under the same assumptions. When the task is to evaluate relational properties of the aggregate relative to private local structure, exponential communication advantages become possible. Understanding which learning primitives fall into each of these regimes is, in our view, a promising direction for future work.

\section{Acknowledgments} This research is partially supported by the National Research Foundation, Singapore through the National Quantum Office, hosted in A*STAR, under its Centre for Quantum Technologies Funding Initiative (S24Q2d0009). MGG is funded by the EPSRC UK Quantum Technologies Programme under grant EP/T001062/1
and VeriQloud. The authors extend their gratitude to the Centre for Quantum Technologies in Singapore for hosting Hackamonth 2026, where this project was developed.

\bibliographystyle{plain}
\bibliography{references}

\section*{Appendix}

\appendix

\section{Superdense coding primer}\label{app:sdc}

Superdense coding~\cite{bw92} is the protocol underlying the factor-of-two bandwidth saving of \Cref{sec:protocol}. It is the mirror image of teleportation: teleportation spends one shared EPR pair and two classical bits to move one qubit, whereas superdense coding spends one shared EPR pair and one transmitted qubit to move two classical bits. We restate it here for readers coming from the distributed-learning side, in the same accessible spirit as the Boolean Hidden Matching primer of \Cref{app:bhm}.

\subsection{The protocol}
Alice and Bob share one EPR pair in the state
\[
  \ket{\Phi^+}=\frac{1}{\sqrt2}\bigl(\ket{00}+\ket{11}\bigr),
\]
with Alice holding the first qubit and Bob the second. Crucially, this pair is distributed in advance, during the input-independent setup phase, before Alice knows which bits she will send. To transmit two classical bits $(a,b)\in\{0,1\}^2$, Alice applies one of four Pauli operations to her half \emph{alone} and then sends that single qubit to Bob. The four operations rotate the shared pair onto the four mutually orthogonal Bell states:
\[
  \begin{array}{c|c|c}
    (a,b) & \text{Alice applies} & \text{resulting Bell state}\\\hline
    00 & I  & \ket{\Phi^+}=(\ket{00}+\ket{11})/\sqrt2\\
    01 & X  & \ket{\Psi^+}=(\ket{01}+\ket{10})/\sqrt2\\
    10 & Z  & \ket{\Phi^-}=(\ket{00}-\ket{11})/\sqrt2\\
    11 & ZX & \ket{\Psi^-}=(\ket{01}-\ket{10})/\sqrt2
  \end{array}
\]
Now holding both qubits, Bob performs a Bell measurement. Because the four states are mutually orthogonal, this measurement distinguishes them perfectly, and Bob recovers $(a,b)$ exactly in the noiseless setting.

\subsection{Why this gives the factor of two}
Only one qubit crosses the channel during the online phase, yet two classical bits are delivered: the entanglement-assisted capacity of a noiseless qubit channel is two classical bits, and superdense coding saturates it~\cite{bsst}. This is exactly the substitution used in \Cref{sec:protocol}, where each classical message of $m$ bits becomes $\lceil m/2\rceil$ transmitted qubits on every ring link.

The saving is not a free lunch in total channel use. Each pair consumed online had to be distributed during setup, at a cost of one qubit channel-use across the same link, so the lifetime carrier count matches the classical baseline, consistent with the Holevo bound. What superdense coding changes is the \emph{timing}: the setup traffic is input-independent and can be scheduled off the synchronisation critical path, leaving only half the volume inside the blocking all-reduce barrier. In the communication-bound regime of large-scale training, that relocation is the operationally relevant gain. 

Classical systems can also exploit input-independent preprocessing, but pre-shared classical resources do not provide an analogue of superdense coding. The factor-of-two reduction in online communication arises specifically from entanglement-assisted communication.

\section{Communication-complexity primer}\label{app:cc}

Communication complexity studies how much two parties must exchange in order to jointly
compute a function of their separate inputs. It is important because in distributed
systems the dominant cost is frequently not local computation but the data moved between
machines, and communication complexity makes that cost precise: it yields lower bounds
that hold no matter how clever the local processing is. The point is sharp in large-scale
machine learning, where training is communication-bound and the exchange of gradients
across devices, rather than the arithmetic on any one of them, sets the pace. A lower
bound on the bits needed for a coordination task is then a statement that no
implementation, present or future, can do better, and a quantum protocol that beats it is
exploiting a genuinely non-classical effect. This is the setting for the separations of
\Cref{sec:tieaudit}, and we collect here the notions our results use.

\subsection{The model}
Two parties, conventionally Alice and Bob, hold inputs $x$ and $y$ and wish to evaluate a function $f(x,y)$. Local computation is free; the only charged resource is the number of bits, or qubits, that they send to one another. The communication complexity of
$f$ is the least communication, over all protocols, that outputs $f(x,y)$ correctly. In our setting the inputs are gradient-derived. In \textsc{GapIP} (\Cref{def:gapip}) Alice holds $g_j$, Bob holds $g_k$, and they decide whether the two gradients are aligned or
conflicting; the point is that they do so \emph{without} exchanging the gradients themselves, since the naive solution sends a full $P$-dimensional vector ($Pb$ bits) and \textsc{GapIP} asks how much smaller a message can be while still resolving the alignment
question, which is the regime where a client must act on a gradient it cannot afford to receive in full. In \textsc{TieAudit} (\Cref{def:tieaudit}) the server holds the sign-gradient $s$ and the client holds a private matching. A \emph{matching} in this setting is a pairing of the $P$ coordinate indices into $P/2$ disjoint pairs $(i,j)$, each pair marking two parameters the client's own model treats as coupled, for instance weights shared across a convolutional filter, or two entries forced equal by a factorisation or a quantisation bin. It is \emph{private} because it is fixed by the client's local structure and is never sent to the server, so the server's message must be useful for whatever pairing the client happens to hold.

\subsection{One-way versus interactive}
A protocol is \emph{one-way} when it consists of a single message from one party to the other, who then produces the answer, and \emph{interactive} when the parties exchange several messages, each able to depend on what was received so far. The distinction is not
cosmetic. \textsc{TieAudit} is one-way: the server sends a message and the client decides, with no return channel, mirroring a server that broadcasts to bandwidth-limited clients. \textsc{GapIP} is solved interactively: the phase-estimation protocol of \Cref{app:gapip} shuttles a register back and forth, and \Cref{rem:margin} shows that the same precision is
unattainable in a single round.

\subsection{Randomness and shared resources}
If the parties must always be correct the problem is \emph{deterministic}; if they may err with probability at most $1/3$, a threshold that repetition can drive down, it is
\emph{randomised}. A randomised protocol may use \emph{shared randomness}: a common random string, fixed independently of the inputs, that both parties can read. Shared randomness is
a free, input-independent correlation, and it is essential for fairly bounding quantum advantage claims, since a quantum advantage should not originate from the mere presence of randomness but from our ability to exploit quantum-mechanical laws, such as quantum
correlations, to outperform classical methods. The classical \textsc{GapIP} protocol of \Cref{thm:margin} uses randomness directly: both workers draw the same random projection matrix $R$ from shared randomness and compute the sketches $Rg_j$ and $Rg_k$, which would not agree if each worker sampled $R$ privately.

The quantum counterpart of shared randomness, in our setting, is \emph{pre-shared entanglement}, a supply of EPR pairs distributed before any input arrives (the multi-party analogue, used by the aggregation protocols of \Cref{sec:protocol} and \Cref{sec:security},
is a shared GHZ state). Throughout \Cref{sec:tieaudit} we use \emph{symmetric resources}: classical protocols may use shared randomness and quantum protocols may use pre-shared
entanglement, both input-independent. This is what makes the comparison fair. A gap obtained by giving the quantum protocol a pre-shared correlation while denying the classical protocol any would be uninformative, because it could merely reflect that one side has a shared resource and the other does not. We instead equip each side with the strongest input-independent correlation of its own kind, shared randomness classically and entanglement quantumly, so that any remaining separation is attributable to quantum mechanics itself rather than to an asymmetry in free resources.

\subsection{Qubits are not free}
A quantum protocol sends qubits in place of bits. One might expect a qubit to carry unboundedly more than a bit, but it cannot: by Holevo's bound a single qubit conveys at most one classical bit of accessible information, and at most two even with unlimited
pre-shared entanglement, the latter saturated by superdense coding~\cite{bw92,bsst}. This ceiling is precisely the factor of two of \Cref{sec:protocol}. The advantages of \Cref{sec:tieaudit} therefore cannot come from packing more data onto each carrier. They
come from a different mechanism: the quantum message lets the receiver extract, through its own choice of measurement, an answer that no equally short classical message could have encoded in advance. This is also the mechanism behind the broader claim of the paper. The quantum communication layer introduced for the all-reduce in \Cref{sec:protocol} does more than complete aggregation at lower online cost: the same layer, by letting receivers choose what to extract, enables capabilities with no classical counterpart at any
communication cost, namely the information-theoretic privacy of \Cref{sec:security} and the conflict-detection separations of the present section.

\subsection{Separations, with examples}
A \emph{separation} compares the quantum cost $Q(f)$ with the classical cost $R(f)$ as functions of the input size. Two regimes appear in this paper.

A \emph{quadratic} separation has $Q(f)=\widetilde\Theta(\sqrt{R(f)})$. The canonical
example is Gap-Hamming Distance, in which the two parties decide whether their $n$-bit
strings are close or far in Hamming distance: classically this requires $\Theta(n)$
bits~\cite{chakrabartiregev}, whereas $\widetilde O(\sqrt{n})$ qubits suffice, a quadratic
gap. \textsc{GapIP} inherits exactly this separation through the reduction in
\Cref{thm:margin}.

An \emph{exponential} separation has the classical cost polynomially large while the
quantum cost is only polylogarithmic. The canonical example is Boolean Hidden Matching
(\Cref{app:bhm}), where any classical one-way protocol needs $\Omega(\sqrt{P})$ bits but
$O(\log P)$ qubits suffice. \textsc{TieAudit} inherits this separation. To make the gap
concrete, a model with $P=10^9$ parameters forces on the order of $3\times10^4$ classical
bits per client, while the quantum message is a handful of states on
$\lceil\log_2 P\rceil\approx 30$ qubits each.

These two examples also display the role of interaction noted in \Cref{sec:tieaudit}: the
quadratic \textsc{GapIP} advantage needs an interactive protocol, whereas the exponential
\textsc{TieAudit} advantage already holds for a single one-way message.

\begin{note}[On terminology]
We speak of \emph{separations} rather than \emph{advantages} deliberately. A separation is a proven gap between the quantum and classical communication complexity of a fixed task,
resting on a lower bound for the classical side, and it is a worst-case, asymptotic statement about communication alone. The word ``advantage'' is looser and is often read as an end-to-end practical benefit. Keeping to ``separation'' makes precise what we claim and what we do not: the quantum protocols provably move less information, but, as \Cref{rem:tieaudit} stresses, this does not by itself establish a faster or better training procedure. What it does establish is a reduction in communication, the metric our ring all-reduce protocol is built to optimise; which in itself is a plausible route to more efficient training.
\end{note}

\section{Boolean Hidden Matching primer}\label{app:bhm}

Boolean Hidden Matching is the problem underlying the exponential separation of \Cref{thm:tie}. It was introduced by Gavinsky et al.~\cite{gavinsky2008exponential}, with an alternative construction by Bar-Yossef, Jayram, and Kerenidis~\cite{baryossef2004exponential},
as the first example of an exponential gap between quantum and classical one-way communication. We state it here in an accessible manner to explain both why it is hard classically and why a small quantum message solves it, building on the one-way model of \Cref{app:cc}.

\subsection{The problem}
Alice holds a bit string $x\in\{0,1\}^P$. Bob holds a \emph{perfect matching} $M$ on the
index set $[P]=\{1,\dots,P\}$, that is, a set of $P/2$ disjoint pairs $(i,j)$ that together
cover every index once, together with a target vector $w\in\{0,1\}^{P/2}$, one bit per
pair. The promise is that the \emph{edge-parity vector}
\[
  \bigl(x_i\oplus x_j\bigr)_{(i,j)\in M}
\]
equals either $w$ or its bitwise complement $\bar w$. From a single message sent by Alice,
Bob must decide which. The only thing that matters about $x$ is, for each of Bob's pairs,
the parity $x_i\oplus x_j$; but Alice does not know $M$, so she must commit her message
before learning which pairs Bob will read.

For a concrete instance, take $P=4$ and $x=(0,0,1,1)$. The matching $M=\{(1,2),(3,4)\}$
yields parities $(x_1\oplus x_2,\,x_3\oplus x_4)=(0,0)$, while the matching
$M'=\{(1,3),(2,4)\}$ on the same $x$ yields $(x_1\oplus x_3,\,x_2\oplus x_4)=(1,1)$. The
parities Bob extracts depend on the matching he holds, and Alice cannot tailor her message
to a matching she has not seen.

\subsection{Why it is hard classically}
Alice's message is a function of $x$ alone, fixed before $M$ is revealed. To answer
correctly for whatever matching Bob holds, it would have to preserve the parity
$x_i\oplus x_j$ of an arbitrary disjoint pairing. There are about $P^2/2$ candidate pairs
but only $P/2$ in any one matching, chosen adversarially after the message is sent, so any
short summary of $x$, a random sketch or a subset of its bits, will miss the pairs Bob
happens to ask about. Making this precise, Gavinsky et al.~\cite{gavinsky2008exponential}
show that any classical one-way protocol with constant advantage must send $\Omega(\sqrt P)$
bits. The threshold sits at $\sqrt P$ rather than $P$ because a message of that size already
starts to cover pairs by a birthday-type coincidence; below it, almost every matching is
effectively unseen.

\subsection{Why a small quantum message suffices}
Alice instead sends the \emph{phase state}
\[
  \ket{x} \;=\; \frac{1}{\sqrt P}\sum_{c=1}^{P} (-1)^{x_c}\,\ket{c},
\]
a single register of $\lceil\log_2 P\rceil$ qubits that records all $P$ bits at once in its
signs. Bob, holding $M$, measures in the basis adapted to his matching, namely
$\bigl\{(\ket i+\ket j)/\sqrt2,\ (\ket i-\ket j)/\sqrt2\bigr\}_{(i,j)\in M}$. Restricted to
one pair $(i,j)$ the state is $\bigl((-1)^{x_i}\ket i+(-1)^{x_j}\ket j\bigr)/\sqrt P$, and
its overlap with $(\ket i-\ket j)/\sqrt2$ is
\[
  \frac{(-1)^{x_i}-(-1)^{x_j}}{\sqrt{2P}},
\]
which vanishes exactly when $x_i=x_j$. The measurement therefore returns a uniformly random
pair $(i,j)\in M$ together with its parity, read off with no error: outcome
$(\ket i+\ket j)/\sqrt2$ means $x_i\oplus x_j=0$, and $(\ket i-\ket j)/\sqrt2$ means
$x_i\oplus x_j=1$.

A single measurement decides the problem. If the edge-parity vector is $w$, the
sampled parity agrees with the corresponding bit of $w$; if it is $\bar w$, it disagrees.
Comparing one extracted parity against $w$ identifies the branch with constant advantage,
using one message of $O(\log P)$ qubits, an exponential saving over the classical
$\Omega(\sqrt P)$.

\subsection{From Boolean Hidden Matching to \textsc{TieAudit}}
\Cref{thm:tie} is the estimation form of this problem. Setting $x_c=\mathbf 1[s_c=+1]$ turns
the sign-gradient into a bit string, so that $x_i\oplus x_j=1$ precisely when $s_i\ne s_j$,
that is, when the tied pair $(i,j)$ is sign-inconsistent. The phase state coincides with the
\textsc{TieAudit} state: $(-1)^{x_c}=-s_c$ gives $\ket x=-\ket s$, the same physical state up
to a global sign, and the client's matching-basis measurement is exactly the extraction
above. The one difference is the output. Boolean Hidden Matching asks a single yes/no
question, $w$ versus $\bar w$, which one copy answers, whereas \textsc{TieAudit} asks for the
\emph{fraction} $f_k$ of inconsistent pairs to additive accuracy $\pm\epsilon$. Estimating a
fraction rather than deciding a binary promise is why \textsc{TieAudit} averages over
$O(\epsilon^{-2})$ copies, for $O(\epsilon^{-2}\log P)$ qubits in total, while the
$\Omega(\sqrt P)$ classical bound carries over because an additive-$\tfrac14$ estimate already
decides the underlying branch.

\section{$\textsc{GapIP}_\tau$ Quantum protocol} \label{app:gapip}

\subsection{Encoding}  
Worker $j$ augments its unit vector to
$\hat u=(1,g_j)/\sqrt2\in\mathbb R^{P+1}$ and locally prepares the amplitude-encoded
state $\ket{\hat u}=U_j\ket0$ on $n=\lceil\log_2(P+1)\rceil=O(\log P)$ qubits, where
$U_j$ depends only on $g_j$; worker $k$ likewise prepares $\ket{\hat v}=U_k\ket0$ from
$\hat v=(1,g_k)/\sqrt2$.  The augmentation does two things: $\langle\hat u,\hat v\rangle
=\tfrac12\bigl(1+\langle g_j,g_k\rangle\bigr)\in[0,1]$ is nonnegative, hence a
\emph{monotone} function of the signed inner product (overlap estimation natively
returns the unsigned angle, and the offset restores the sign), and it sends the
decision threshold $\langle g_j,g_k\rangle=\pm\tau$ to $\langle\hat u,\hat v\rangle
=\tfrac12(1\pm\tau)$, in the interior of $[0,1]$ where $\arccos$ is well-conditioned.
Only communication is counted; the local cost of $U_j$ (up to $O(P)$ gates) is not
charged.

\subsection{Local reflections}  
Each reflection is a local operation on the shared
register: $R_{\hat u}=2\ket{\hat u}\bra{\hat u}-I=U_j R_0 U_j^\dagger$ with
$R_0=2\ket0\bra0-I$ a multi-controlled phase, applied by worker $j$ alone; symmetrically
$R_{\hat v}=U_k R_0 U_k^\dagger$ by worker $k$.

\subsection{Distributed Grover operator}  
The two-party operator $G=R_{\hat u}R_{\hat v}$ is
realised on a single travelling $n$-qubit register: worker $k$ applies $R_{\hat v}$,
transmits the register to worker $j$, who applies $R_{\hat u}$ and returns it.  One
application of $G$ (or of its controlled version, carrying an $O(1)$-qubit
phase-estimation ancilla along with the register) therefore costs one round-trip of
$O(\log P)$ qubits.  Within the plane $\mathrm{span}\{\hat u,\hat v\}$, $G$ is a rotation
by $2\theta$ with $\theta=\arccos\langle\hat u,\hat v\rangle$; it acts trivially on the
orthogonal complement, and the initial state $\ket{\hat v}$ lies in the plane.

\subsection{Estimation and decision}  
Iterative (Kitaev) phase estimation~\cite{kitaev2002}
on $G$, run on $\ket{\hat v}$, estimates the eigenphase $2\theta$ (hence
$\langle\hat u,\hat v\rangle=\cos\theta$, hence $\langle g_j,g_k\rangle$) to additive
$\pm\Theta(\tau)$ using $O(1/\tau)$ applications of $G$, the Heisenberg rate
($\cos\theta$ is injective on $\theta\in[0,\pi]$, so the eigenphase sign is immaterial).
Each application is one round-trip, giving $O(\tau^{-1}\log P)$ qubits; boosting the
success probability to $\ge2/3$ multiplies this by $O(\log(1/\delta))$, for
$\widetilde O(\tau^{-1}\log P)$ total.  The protocol is manifestly interactive
($O(1/\tau)$ sequential round-trips), consistent with \Cref{rem:margin}: no one-round
protocol attains the Heisenberg rate.

\end{document}